\documentclass[11pt]{article}

\usepackage{amsmath,amssymb,amsthm,mathtools}
\usepackage[a4paper,margin=1in]{geometry}
\usepackage{cite}
\usepackage{tikz}
\usepackage{aliascnt}
\usepackage{microtype}
\usepackage{enumitem}
\usepackage[hidelinks,bookmarks=false]{hyperref}
\hypersetup{
  pdftitle={Exact Net-Occurrence Counts in Purely Morphic Regular Epistandard Words},
  pdfauthor={Simone Faro and Arianna Pavone}
}
\usepackage[nameinlink,capitalise]{cleveref}

\newtheorem{theorem}{Theorem}[section]

\newaliascnt{lemma}{theorem}
\newtheorem{lemma}[lemma]{Lemma}
\aliascntresetthe{lemma}

\newaliascnt{corollary}{theorem}
\newtheorem{corollary}[corollary]{Corollary}
\aliascntresetthe{corollary}

\newaliascnt{proposition}{theorem}

\aliascntresetthe{proposition}

\theoremstyle{definition}
\newaliascnt{definition}{theorem}

\aliascntresetthe{definition}

\theoremstyle{remark}
\newaliascnt{remark}{theorem}
\newtheorem{remark}[remark]{Remark}
\aliascntresetthe{remark}

\newcommand{\NO}{\mathrm{NO}}
\newcommand{\Occ}{\operatorname{Occ}}
\newcommand{\Pal}{\operatorname{Pal}}

\newcommand{\eps}{\varepsilon}

\begin{document}

\title{Exact Net-Occurrence Counts in Purely Morphic Regular Epistandard Words}

\author{%
Simone Faro\\
\small Department of Mathematics and Computer Science, University of Catania\\
\small Viale Andrea Doria 6, 95125 Catania, Italy\\
\small \texttt{faro@dmi.unict.it}\\
\small ORCID: 0000-0001-5937-5796
\and
Arianna Pavone\\
\small Department of Mathematics and Computer Science, University of Palermo\\
\small Via Archirafi 34, 90123 Palermo, Italy\\
\small \texttt{ariannamaria.pavone@unipa.it}\\
\small ORCID: 0000-0002-8840-6157
}
\date{}

\maketitle

\begin{abstract}
Finite Fibonacci words have recently been shown to contain exactly three net occurrences---occurrences of repeated factors whose one-letter left and right extensions are unique. This unexpectedly small constant raises a natural question: is it peculiar to the Fibonacci recurrence, or part of a broader morphic phenomenon? We study a family over the alphabet $\{0,1,\ldots,d-1\}$ determined by positive integers $e_0,\ldots,e_{d-1}$. For each letter $a$, let $L_a$ be the morphism that fixes $a$ and maps every other letter $b$ to $ab$; we consider the finite approximants $S_m=\mu^m(0)$ generated by $\mu=L_0^{e_0}\cdots L_{d-1}^{e_{d-1}}$. These words are the period-aligned finite approximants of the purely morphic regular epistandard family considered here. We prove a sharp dichotomy: for every $m\ge2$, $S_m$ has exactly three net occurrences when $e_{d-1}=1$, and exactly two when $e_{d-1}\ge2$; the initial approximant is also completely classified. The proof combines palindromic prefixes, return-word factorizations, and overlapping net-occurrence covers. As a consequence, when all exponents are equal to one---the standard $d$-bonacci case---every noninitial period-aligned finite approximant has exactly three net occurrences, placing the Fibonacci phenomenon in a wider epistandard framework.
\end{abstract}



\section{Introduction}

Net occurrences provide a local way of identifying repeated substrings, or \emph{factors}, that are not explained by longer repeated contexts. An occurrence of a repeated factor is \emph{net} when both its one-character left extension and its one-character right extension are unique; occurrences touching the boundary of the finite word are handled by the natural boundary convention formalized in \cref{sec:word-basics}. Guo and Kishi recently proved that finite Fibonacci words, generated by iterating $0\mapsto01$ and $1\mapsto0$, have exactly three net occurrences, whereas finite Thue--Morse words, generated by iterating $0\mapsto01$ and $1\mapsto10$, have exactly nine \cite{GuoKishi2025}. Their analysis introduced the notion of an \emph{overlapping net occurrence cover} (ONOC), a left-to-right chain of overlapping net occurrences that covers the whole finite word and can be used to prove that a prescribed collection of net occurrences is complete. A subsequent morphism-based approach revisited these two families through occurrence-preserving morphisms and gave alternative proofs of their net-occurrence structure \cite{KishiEtAl2026}.

From a combinatorial viewpoint, the Fibonacci result raises a natural structural question: is the constant number of net occurrences peculiar to the binary Fibonacci recurrence, or does it persist in broader families of morphically generated words? A first extension is suggested by the $d$-bonacci family, the natural $d$-letter generalization of the Fibonacci construction. We use the standard $d$-bonacci substitution (letter-to-word morphism) $\sigma$ on $\{0,\ldots,d-1\}$, defined by
\(\sigma(i)=0(i+1)\quad(0\le i<d-1), \qquad \sigma(d-1)=0.\)
Its associated epistandard directive word is
\((01\cdots(d-1))^\omega,\)
as recalled for the $d$-bonacci family in \cite{Peltomaki2024}.
More broadly, an epistandard word is obtained from an infinite \emph{directive word}---an infinite sequence of letters prescribing the successive construction steps---by repeated palindromic closure, which extends the current prefix to the shortest possible palindrome; this construction goes back to the foundational development of episturmian words \cite{DroubayJustinPirillo2001} and is recalled formally in \cref{sec:epi-basics}. This framework contains the Fibonacci, Tribonacci, and $d$-bonacci words as particularly symmetric examples \cite{GlenJustin2009,Peltomaki2024}. It also gives explicit descriptions of palindromic prefixes, factors, and \emph{return words}, the blocks running from one occurrence of a factor up to (but not including) the next occurrence \cite{JustinVuillon2000,GlenJustin2009}.

The notion of net occurrence also has an algorithmic origin. Net frequency was introduced by Lin and Yu \cite{LinYu2001}, and recent work has developed offline \cite{GuoEtAl2024,OhlebuschEtAl2024}, online \cite{GuoEtAlOnline2024,MienoInenaga2025,Inenaga2026}, and compressed-space \cite{KimuraI2026} algorithms for computing net frequency and net occurrences. Structured morphic texts already enter this line of research: Fibonacci words were used in obtaining lower bounds for net-frequency computation \cite{GuoEtAl2024,GuoKishi2025}. This algorithmic background complements the combinatorial problem addressed here, namely the exact structure and number of net occurrences in a natural epistandard family.

In this paper we consider the following periodic family. Let
\(\Sigma_d=\{0,1,\ldots,d-1\},\qquad d\ge2,\)
and fix integers $e_0,e_1,\ldots,e_{d-1}\ge1$. A morphism on words is determined by its images on letters and is extended to arbitrary words by concatenation. For $a\in\Sigma_d$, let $L_a$ be the elementary pure epistandard morphism defined by
\(L_a(a)=a,\qquad L_a(b)=ab\quad(b\ne a).\)
For a finite directive word $v=x_1\cdots x_n$, we write
\(\mu_v=L_{x_1}\cdots L_{x_n}\)
for the corresponding composition, with composition taken from right to left. We study the morphism associated with one directive period
\(\mu=L_0^{e_0}L_1^{e_1}\cdots L_{d-1}^{e_{d-1}}\)
and its finite approximants
\(S_m=\mu^m(0),\qquad m\ge1.\)
We call $S_1$ the \emph{initial approximant} and $S_m$, $m\ge2$, the \emph{noninitial approximants}. All net-occurrence counts in this paper refer to these canonical finite approximants $S_m$; no claim is made about arbitrary finite factors of the associated infinite word. Since the same morphism $\mu$ is iterated at every step, the infinite limit is called \emph{purely morphic}. Equivalently, the associated epistandard word has periodic directive word
\(\Delta=\left(0^{e_0}1^{e_1}\cdots(d-1)^{e_{d-1}}\right)^\omega,\)
where $v^\omega$ denotes infinite repetition of the finite word $v$. After equal consecutive directive letters are grouped into runs, their positive lengths are traditionally called \emph{partial quotients}. Thus the run letters cycle through $0,1,\ldots,d-1$ and the corresponding run lengths repeat every $d$ runs; this is the regular cyclic setting used throughout the paper. Periodic directive words are a classical source of purely morphic epistandard words \cite{JustinPirillo2002,GlenJustin2009}, while Peltom\"aki developed the associated \emph{generalized standard words}, canonical finite words attached to the successive run boundaries, in detail \cite{Peltomaki2024}.

The recent occurrence-preserving framework of Kishi, Guo, Urbina, and Bannai \cite{KishiEtAl2026} provides an important comparison point, but it does not globally subsume the morphisms considered here. To see this directly, put $z=d-1$ and $\nu=L_0^{e_0}\cdots L_{z-1}^{e_{z-1}}$. Then
\(\mu(z)=\nu(z),\qquad \mu(0)=\nu(z)^{e_z}\nu(0).\)
Thus $z$ has no occurrence in the one-letter word $0$, whereas $\mu(z)$ occurs in $\mu(0)$. Each elementary morphism $L_a$ is injective \cite{JustinPirillo2002}, and hence so is $\mu$. Proposition~23 of Kishi et al. states that, for an injective morphism, interference-freeness on a factor $u$ implies preservation of the number of occurrences of $u$ after one morphic step \cite[Proposition~23]{KishiEtAl2026}. Taking $u=z$ and $v=0$, the strict increase from zero occurrences of $z$ in $0$ to at least one occurrence of $\mu(z)$ in $\mu(0)$ therefore shows that $\mu$ is not interference-free even on the singleton factor $z$. Their machinery may still apply to specially chosen factors, but it does not globally subsume the present setting; our classification instead relies on factor-specific localization inside the epistandard approximants, supplied by palindromic containers and return words.

The restriction that the same exponent vector $(e_0,\ldots,e_{d-1})$ repeats at every directive period is also structural rather than merely notational. It makes the same period morphism $\mu$ recur at every level and keeps the relevant return-word factorizations aligned with fixed run positions. If the partial quotients vary from one cycle to the next, the corresponding palindromic containers and finite derived words depend on the whole preceding exponent sequence, so the two- or three-occurrence ONOC geometry is no longer forced by a single terminal exponent. The present setting is therefore a natural homogeneous purely morphic subclass in which a uniform exact count can be sought.

Our main result is a sharp dichotomy. We write $\NO(T)$ for the set of net occurrences of a finite word $T$, and $\#X$ for the number of elements of a finite set $X$. For every noninitial approximant, the number of net occurrences depends only on the exponent of the last letter in one directive period:
\[
\#\NO(S_m)=
\begin{cases}
3,&e_{d-1}=1,\\
2,&e_{d-1}\ge2,
\end{cases}
\qquad m\ge2.
\]
We also determine the initial approximant completely. In full form,
\[
\#\NO(S_m)=
\begin{cases}
3,&e_{d-1}=1\text{ and }(m\ge2\text{ or }e_0\ge2),\\
2,&\text{otherwise}.
\end{cases}
\]
Hence the entire family exhibits a two-or-three dichotomy controlled, after the initial approximant, by a single partial quotient.

The proof combines three structural ingredients. First, generalized standard words and palindromic prefixes locate the factors that can support net occurrences. Second, the return-word theory of Justin and Vuillon determines the positions and immediate left/right contexts of these factors inside the finite approximants \cite{JustinVuillon2000}. Third, the resulting canonical net occurrences are organized into an ONOC. The completeness criterion of Guo and Kishi then reduces the exclusion of every additional net occurrence to the uniqueness of a small number of expanded overlaps \cite{GuoKishi2025}. All terminology used in this outline is defined in \cref{sec:preliminaries} before it enters the proofs.

The two branches of the dichotomy have different geometries. When $e_{d-1}\ge2$, two canonical boundary occurrences overlap and already form an ONOC. When $e_{d-1}=1$, the two boundary occurrences no longer suffice: a canonical palindromic factor gives rise to a third, central net occurrence. For alphabet size $d=2$, the epistandard setting is the classical standard Sturmian setting \cite{GlenJustin2009}; this binary case requires a slightly different factorization, but the same return-word mechanism yields the identical dichotomy.

When $e_0=\cdots=e_{d-1}=1$, the one-period morphism satisfies $\mu=\sigma^d$ for the $d$-bonacci substitution defined above. Hence our theorem yields exactly three net occurrences on the period-aligned subsequence $S_m=\sigma^{dm}(0)$ for $m\ge2$, that is, on every $d$th substitution iterate. This qualification is essential: the intermediate iterates $\sigma^{dm+r}(0)$, $1\le r<d$, stop inside a directive period and are not covered by the main theorem, so no claim about their net-occurrence counts is made here. In the binary case the covered subsequence is the even-indexed Fibonacci subsequence; the full Fibonacci statement is Theorem~34 of Guo and Kishi \cite{GuoKishi2025}.

\section{Preliminaries}
\label{sec:preliminaries}

\subsection{Words, occurrences, contexts, and periods}
\label{sec:word-basics}

Let $\Sigma$ be a finite alphabet and let $\Sigma^*$ denote the set of all finite words over $\Sigma$; $\eps$ denotes the empty word. For a finite word $w$, its length is $|w|$. For an integer $q\ge0$, the power $w^q$ is the concatenation of $q$ copies of $w$, while $w^\omega$ denotes the right-infinite periodic word $www\cdots$ when $w\ne\eps$.

A word $v$ is a \emph{factor} of $w$ if $w=xvy$ for some words $x,y$. It is a \emph{prefix} if $x=\eps$ and a \emph{suffix} if $y=\eps$; it is \emph{proper} if it is different from the whole word. We write $v\prec w$ when $v$ is a proper prefix of $w$. A \emph{palindrome} is a word that reads the same from left to right and from right to left. If $x$ is a prefix of $w$, the notation $x^{-1}w$ means that this prefix is deleted; if $x$ is a suffix of $w$, then $wx^{-1}$ means that the suffix is deleted. These are quotient notations for words, not group inverses.

For a finite word $T\in\Sigma^*$ we use one-based indexing. If $1\le i\le j\le|T|$, the interval $(i,j)$ denotes the occurrence of the factor $T[i..j]$. When a named factor $W$ starts at position $i$, we also write $W@i$, and we denote by
\[
\Occ_T(W)=\{i:1\le i\le |T|-|W|+1,\ T[i..i+|W|-1]=W\}
\]
the set of all starting positions of occurrences of $W$ in $T$. An occurrence is \emph{repeated} if the same factor occurs at a different position in $T$. An occurrence $(i',j')$ is a \emph{sub-occurrence} of $(i,j)$ if $i\le i'\le j'\le j$, and is proper when at least one inequality is strict; \emph{super-occurrence} is the reverse relation. A \emph{complete occurrence} means that the entire factor lies inside the finite word under consideration.

When both neighboring positions exist, the \emph{context} of an occurrence $(i,j)$ is the ordered pair $(T[i-1],T[j+1])$. We use the boundary symbol $\bot$ in place of a missing neighbor; for example, a prefix occurrence with right neighbor $b$ has context $(\bot,b)$.

Following Guo and Kishi \cite{GuoKishi2025}, an occurrence $(i,j)$ is a \emph{net occurrence} if $T[i..j]$ is repeated, whereas its one-character left and right extensions are unique. More precisely, $T[i-1..j]$ and $T[i..j+1]$ must each occur exactly once whenever they are defined. If $i=1$, the left extension is regarded as unique by convention; similarly, if $j=|T|$, the right extension is regarded as unique. We write $\NO(T)$ for the set of net occurrences of $T$.

We use the following elementary properties from \cite{GuoKishi2025}: a proper sub-occurrence of a net occurrence is not net, and a proper super-occurrence of a net occurrence is not net.

A positive integer $p\le |w|$ is a \emph{period} of $w$ if $w[i]=w[i+p]$ whenever both positions are defined. The smallest such $p$ is the \emph{least period}. A nonempty word is \emph{primitive} if it is not a proper power $v^q$ with $q\ge2$. For a letter $a$, $|w|_a$ denotes the number of occurrences of $a$ in $w$; the vector $(|w|_a)_{a\in\Sigma}$ is the \emph{Parikh vector} of $w$. For positive integers $p$ and $q$, $\gcd(p,q)$ denotes their greatest common divisor.

\subsection{Overlapping net occurrence covers}

Let $\mathcal C=\{(i_1,j_1),\ldots,(i_c,j_c)\}$ be a collection of net occurrences ordered from left to right. Following \cite{GuoKishi2025}, $\mathcal C$ is an \emph{overlapping net occurrence cover} (ONOC) of $T$ if
\(i_1=1,\qquad i_{k+1}\le j_k\ (1\le k<c),\qquad j_c=|T|.\)
For each $1\le k<c$, the overlap $(i_{k+1},j_k)$ is called a \emph{bridging net sub-occurrence} (BNSO); its one-character expansion is $(i_{k+1}-1,j_k+1)$.

We repeatedly use the following consequence of \cite[Lemma~12]{GuoKishi2025}: if every expanded BNSO of an ONOC occurs only once in $T$, then that ONOC already contains all net occurrences of $T$.

\subsection{Epistandard morphisms, directive words, and palindromic closure}
\label{sec:epi-basics}

For finite alphabets $\Sigma$ and $\Gamma$, a \emph{morphism} $h:\Sigma^*\to\Gamma^*$ satisfies $h(uv)=h(u)h(v)$ for all words $u,v$; hence it is determined by its images on individual letters. A morphism from $\Sigma^*$ to itself is also called a \emph{substitution}. For $a\in\Sigma$, the elementary pure epistandard morphism $L_a$ is defined by
\(L_a(a)=a,\qquad L_a(b)=ab\quad(b\ne a).\)
Compositions of the $L_a$ are called \emph{pure epistandard morphisms} \cite{JustinPirillo2002,GlenJustin2009}. Products denote composition from right to left:
\(L_{a_1}L_{a_2}(w)=L_{a_1}(L_{a_2}(w)).\)
For a finite word $v=x_1\cdots x_n$, we write
\(\mu_v=L_{x_1}\cdots L_{x_n}.\)

For later use, the \emph{incidence matrix} $M_h$ of a morphism $h:\Sigma^*\to\Sigma^*$ is the matrix whose entry $(a,b)$ is $|h(b)|_a$. A square integer matrix is \emph{unimodular} if its determinant is $+1$ or $-1$; equivalently, its inverse has integer entries.

A \emph{directive word} is an infinite sequence $\Delta=\xi_1\xi_2\cdots$ of letters. Set $u_1=\eps$ and recursively
\(u_{n+1}=(u_n\xi_n)^{(+)},\)
where $w^{(+)}$ is the shortest palindrome having $w$ as a prefix. This recursion defines the \emph{iterated palindromic-closure map} $\Pal$ by
\(u_{n+1}=\Pal(\xi_1\cdots\xi_n).\)
The words $u_n$ form an increasing sequence of palindromic prefixes; their infinite limit is the \emph{epistandard word directed by $\Delta$} \cite{GlenJustin2009}.

A \emph{directive cut} after a finite prefix $v$ of $\Delta$ is simply the split $\Delta=v\Delta'$. The associated morphism is $\mu_v$, and $\Delta'$ is the \emph{directive tail}. We use this language repeatedly when return words are read from a finite prefix of the directive sequence.

Throughout the paper,
\[
\Delta=\delta^\omega,\qquad \delta=0^{e_0}1^{e_1}\cdots(d-1)^{e_{d-1}},
\]
and
\[
\mu=L_0^{e_0}L_1^{e_1}\cdots L_{d-1}^{e_{d-1}},\qquad S_m=\mu^m(0).
\]

\subsection{Regular directive words and generalized standard words}

After equal consecutive letters of a directive word are grouped into maximal \emph{runs}, write a regular cyclic directive word in the form
\(\Delta=y_1^{a_1}y_2^{a_2}\cdots, \qquad r_k=a_1+\cdots+a_k,\)
where the run letters cycle as
\(y_1y_2\cdots=(01\cdots(d-1))^\omega.\)
The positive run lengths $a_k$ are the \emph{partial quotients}, and $r_k$ is the directive position at the end of the $k$th run. Define
\(\tau_k=L_{y_1}^{a_1}\cdots L_{y_k}^{a_k},\qquad s_k=\tau_k(y_{k+1}).\)
The words $s_k$ are the \emph{generalized standard words} associated with the directive runs. We use the structural results of Peltom\"aki \cite[Lemmas~3.15 and~3.19]{Peltomaki2024}: $u_{r_k}$ is a proper prefix of $s_k$, and if the suffix $t_k$ is defined by $u_{r_k+1}=s_kt_k$, then
\(u_{r_k+i}=s_k^it_k\qquad(1\le i\le a_{k+1}).\)
We also invoke Peltom\"aki's explicit product formulas and recurrence for these words; the relevant lemma is cited at each use.

\subsection{Return words, derived words, and return boundaries}

An infinite word is \emph{recurrent} if every one of its finite factors occurs infinitely many times. Let $v$ be a factor of a recurrent infinite word. A \emph{return word} to $v$ starts at one occurrence of $v$ and ends immediately before the next occurrence. The corresponding \emph{complete return} is obtained by adjoining that next occurrence of $v$, so it runs from the beginning of one occurrence of $v$ through the end of the next.

Choose one code letter for each distinct return word and let $\rho(a)$ denote the return word of type $a$. The morphism $\rho$ is the \emph{return-word morphism}. Reading the successive return types gives the \emph{derived word}. Positions at which consecutive return words meet are \emph{return boundaries}; the code letter of a return word is its \emph{return type}. For a finite concatenation $T=\rho(y_1\cdots y_q)$, the finite word $y_1\cdots y_q$ is the corresponding \emph{finite derived word}. This terminology lets us translate questions about positions and neighboring letters in $T$ into short patterns in the derived word.

Return words in epistandard words admit an explicit description due to Justin and Vuillon \cite[Theorem~3.3, Corollary~4.1, and Theorem~4.4]{JustinVuillon2000}. We call the shortest palindromic prefix containing a factor $v$ its \emph{palindromic container}. If
\(u_{n+1}=fvg\)
is this palindromic container, and $\mu_n=L_{\xi_1}\cdots L_{\xi_n}$, then the return words to $v$ are
\(f^{-1}\mu_n(x)f\)
for the letters $x$ occurring in the directive tail. We shall use this theorem together with finite truncations of the corresponding return-word factorization.

\section{Structural and Return-Word Tools}
\label{sec:tools}

This section isolates the two pieces that recur throughout the proof. The word $A_m$ is the repeated block created by the last directive run, whereas $C_m$ is the terminal factor that survives at the end of the approximant. Locating all occurrences of $C_m$ is the key step that later turns local return-word information into a complete count of net occurrences.

Put $z=d-1$ and $r=e_z$. Define
\(\nu=L_0^{e_0}\cdots L_{z-1}^{e_{z-1}},\qquad C=\nu(0),\qquad A=\nu(z),\)
and
\(C_m=\mu^{m-1}(C),\qquad A_m=\mu^{m-1}(A).\)

Before invoking return words, it is useful to make the repeated block structure of $S_m$ completely explicit. The next lemma shows that the whole approximant consists of $r$ copies of $A_m$ followed by $C_m$, and that $C_m$ is already nested at the beginning of $A_m$; this simple geometry drives all later overlap arguments.

\begin{lemma}[Structural decomposition]
\label{lem:structural}
For every $m\ge1$,
\(S_m=A_m^rC_m.\)
Moreover,
\(A=C0^{e_0-1}z,\)
so $C_m\prec A_m$ for every $m$.
\end{lemma}

\begin{proof}
Let $\eta=L_1^{e_1}\cdots L_{z-1}^{e_{z-1}}$ (the identity if $z=1$). There exists a word $P$ such that $\eta(0)=P0$ and $\eta(z)=Pz$. Applying $L_0^{e_0}$ gives
\(C=L_0^{e_0}(P)0, \qquad A=L_0^{e_0}(P)0^{e_0}z=C0^{e_0-1}z.\)
Since $L_z^r(0)=z^r0$,
\(\mu(0)=\nu(z)^r\nu(0)=A^rC,\)
and applying $\mu^{m-1}$ proves the claim.
\end{proof}

At the directive cut immediately before the final run $z^r$ in the $m$th period, set
\(\lambda_m=\mu^{m-1}\nu.\)
Then
\(A_m=\lambda_m(z),\qquad C_m=\lambda_m(0),\qquad S_m=\lambda_m(z^r0).\)
The return-word description at this cut, together with the generalized-standard-word localization of $C_m$, gives the following occurrence statement. Intuitively, the lemma says that there are no hidden copies of $C_m$ inside an $A_m$ block: every occurrence sits exactly at one of the visible block boundaries. This exact localization is what later lets us identify the boundary net occurrences and their overlaps, including the borderline case $r=2$.

\begin{lemma}[Occurrences of the terminal factor]
\label{lem:Cocc}
Assume $d\ge3$. For every $m\ge1$,
\[
\Occ_{S_m}(C_m)=\{1,|A_m|+1,2|A_m|+1,\ldots,r|A_m|+1\}.
\]
\end{lemma}

\begin{proof}
Let $K=md-1$ be the run index of the cut immediately before the final run $z^r$ in the $m$th period. Thus $\lambda_m=\tau_K$, $A_m=s_K$, and $C_m=\tau_K(0)$. We first identify the shortest palindromic prefix that contains $C_m$.

Assume first that $m\ge2$. Applying Peltom\"aki's explicit image formula for regular directive words \cite[Lemma~5.9(v)]{Peltomaki2024} at the last completed run gives
\[
C_m=s_{K-1}^{a_K}R_K,
\qquad
R_K=s_{K-2}^{a_{K-1}}\cdots
s_{K-d+2}^{a_{K-d+3}}s_{K-d+1},
\tag{3.1a}
\]
Here and below, products over an empty index range are simply omitted.

By \cite[Lemma~3.19]{Peltomaki2024},
\[
u_{r_K}=s_{K-1}^{a_K}t_{K-1}.
\tag{3.1b}
\]
In the cyclic sequence of directive runs, the previous occurrence of the directive letter at run $K-1$ is at run $K-d$. We spell out the short cancellation that identifies the suffix $t_{K-1}$. Peltom\"aki's Eqs.~(5) and~(7) imply, for every run index $j\ge1$,
\[
u_{r_j}=s_{j-1}^{a_j-1}s_{j-2}^{a_{j-1}}\cdots s_0^{a_1}.
\]
Indeed, Eq.~(7) gives the same product with one additional leading copy of $s_{j-1}$ for $u_{r_j+1}$, and Eq.~(5) states that $u_{r_j+1}=s_{j-1}u_{r_j}$. On the other hand, the regular recurrence \cite[Eq.~(13)]{Peltomaki2024} gives
\[
s_{K-1}=s_{K-2}^{a_{K-1}}\cdots s_{K-d}^{a_{K-d+1}}s_{K-d-1}.
\]
Expanding $u_{r_{K-1}+1}=s_{K-1}t_{K-1}$ by Eq.~(7), substituting this expression for $s_{K-1}$, and cancelling the common prefix leaves
\[
t_{K-1}=s_{K-d-1}^{a_{K-d}-1}s_{K-d-2}^{a_{K-d-1}}\cdots s_0^{a_1}=u_{r_{K-d}}.
\tag{3.1c}
\]
By \cite[Lemma~3.15]{Peltomaki2024}, $u_{r_{K-d}}$ is a proper prefix of $s_{K-d}$. The generalized-standard-word lengths are strictly increasing by the positive length recurrences in \cite[Eqs.~(10)--(11)]{Peltomaki2024}; since $d\ge3$, we therefore have
\(|t_{K-1}|<|s_{K-d}|<|s_{K-2}|\le |R_K|.\)
Together with (3.1a)--(3.1b), this proves
\(|u_{r_K}|<|C_m|.\)
Since $C_m\prec A_m=s_K$ by \cref{lem:structural}, while Lemma~3.19 gives $u_{r_K+1}=s_Kt_K$, we also have
\(C_m\prec u_{r_K+1}.\)
Consequently $u_{r_K+1}$ is the first palindromic prefix containing $C_m$.

It remains to justify the same localization when $m=1$. Here $K=d-1$. Peltom\"aki's first-cycle image formula \cite[Lemma~5.8]{Peltomaki2024} gives
\[
C_1=s_{K-1}^{a_K}s_{K-2}^{a_{K-1}}\cdots s_1^{a_2}s_0.
\]
Writing $q_i=|s_i|$, the length formula used in the proof of \cite[Lemma~3.15]{Peltomaki2024} gives
\[
|u_{r_K}|=(a_K-1)q_{K-1}+a_{K-1}q_{K-2}+\cdots+a_1q_0.
\]
Therefore
\(|C_1|-|u_{r_K}|=q_{K-1}-(a_1-1)q_0>0,\)
since $q_0=1$ and $q_{K-1}\ge a_1+1$. Moreover, Peltom\"aki's product formula for palindromic prefixes \cite[Eq.~(7)]{Peltomaki2024} gives
\[
u_{r_K+1}=s_{K-1}^{a_K}\cdots s_0^{a_1}.
\]
The first-cycle recurrence for the generalized standard words \cite[Lemma~3.3]{Peltomaki2024} and \cref{lem:structural} identify this word as
\(u_{r_K+1}=s_Kz^{-1}=C_1 0^{e_0-1},\)
where the notation simply means that the final $z$ of $s_K=C_1 0^{e_0-1}z$ is deleted. Thus again $u_{r_K+1}$ is the first palindromic prefix containing $C_1$.

Let $Q_m=u_{r_K+1}$. In both cases $C_m$ is a prefix of $Q_m$. By Justin--Vuillon's Theorem~3.3, the occurrences of $C_m$ are in bijection with the occurrences of its minimal palindromic container $Q_m$ \cite[Theorem~3.3]{JustinVuillon2000}. Writing $Q_m=fC_mg$, here $f=\eps$ because $C_m$ is a prefix of $Q_m$; Corollary~4.1 therefore transfers return words from $Q_m$ to $C_m$ without a conjugating prefix, and Theorem~4.4 identifies them as $\lambda_m(x)$ for the letters $x$ in the directive tail \cite[Corollary~4.1 and Theorem~4.4]{JustinVuillon2000}. Since
\(S_m=\lambda_m(z^r0),\)
the finite derived word is $z^r0$. Hence the complete occurrences of $C_m$ in $S_m$ start exactly at
\(1,\ |A_m|+1,\ 2|A_m|+1,\ldots,r|A_m|+1.\)
There are no intermediate occurrences because every occurrence of $C_m$ corresponds to one of these return boundaries. The restriction $d\ge3$ is essential: in the binary case the above localization changes, and that case is handled separately below.
\end{proof}

We next record the more delicate localization used when $r=1$ and $d\ge3$. In this branch the periodic block argument is no longer sufficient, so we must know exactly which palindromic prefix controls the occurrences of the shorter factor $Z_m$. The next lemma identifies that container and isolates the only exceptional ternary configuration. Let $w_m$ be $\delta^m$ with its final letters $z-1,z$ removed, let $\lambda'_m$ be its directive morphism, and put
\(Z_m=\lambda'_m(0),\qquad U_m=\Pal(w_m).\)

\begin{lemma}[Localization]
\label{lem:localization}
Let $m\ge2$, $d\ge3$, and $e_z=1$. If $d\ge4$ or $e_{z-1}\ge2$, then $U_m$ is the shortest palindromic prefix containing $Z_m$. If $d=3$ and $e_1=1$, the shortest palindromic prefix containing $Z_m$ is $\Pal(\delta^{m-1})$.
\end{lemma}

\begin{proof}
Use the run notation $r_k,s_k,t_k$ and put $k=md-2$. Set $R_k=\tau_k(0)$. Peltom\"aki's explicit image formula for regular epistandard words \cite[Lemma~5.9(v)]{Peltomaki2024} gives
\[
R_k=
 s_{k-1}^{a_k}s_{k-2}^{a_{k-1}}\cdots
 s_{k-d+3}^{a_{k-d+4}}s_{k-d+2},
\tag{3.1}
\]
With the empty-range convention stated above, comparing this expression with Peltom\"aki's recurrence for the generalized standard words $s_k$ \cite[Lemma~3.3]{Peltomaki2024} gives
\[
R_k\prec s_k.
\tag{3.2}
\]
On the other hand, the same cancellation used above makes the suffix in $u_{r_k+1}=s_kt_k$ explicit. Equation~(7) expands $u_{r_k+1}$ as
\[
s_{k-1}^{a_k}s_{k-2}^{a_{k-1}}\cdots s_0^{a_1},
\]
whereas the regular recurrence \cite[Eq.~(13)]{Peltomaki2024} expands
\[
s_k=s_{k-1}^{a_k}\cdots s_{k-d+1}^{a_{k-d+2}}s_{k-d}.
\]
After cancelling this common prefix, the remaining suffix is
\[
s_{k-d}^{a_{k-d+1}-1}s_{k-d-1}^{a_{k-d}}\cdots s_0^{a_1},
\]
which is exactly $u_{r_{k-d+1}}$ by Eqs.~(5) and~(7). Hence
\[
t_k=u_{r_{k-d+1}}.
\tag{3.3}
\]
Since $u_{r_j}$ is a proper prefix of $s_j$, (3.1)--(3.3) imply
\[
|t_k|<|R_k|<|s_k|.
\tag{3.4}
\]

If $e_{z-1}\ge2$, put $\ell=e_{z-1}-1$. Then
\[
Z_m=s_k^\ell R_k,
\qquad
u_{r_k+\ell}=s_k^\ell t_k,
\qquad
u_{r_k+\ell+1}=s_k^{\ell+1}t_k.
\]
By (3.4),
\(|u_{r_k+\ell}|<|Z_m|<|u_{r_k+\ell+1}|,\)
so $u_{r_k+\ell+1}=U_m$ is the shortest palindromic prefix containing $Z_m$.

If $e_{z-1}=1$ and $d\ge4$, then $Z_m=R_k$. Here the length subtraction can also be made explicit. Write $q_i=|s_i|$. From (3.1),
\[
|R_k|=a_kq_{k-1}+a_{k-1}q_{k-2}+\cdots+a_{k-d+4}q_{k-d+3}+q_{k-d+2}.
\]
Using Eqs.~(5) and~(7) as above, the run-end central word satisfies
\[
\begin{aligned}
|u_{r_k}|={}&(a_k-1)q_{k-1}+a_{k-1}q_{k-2}+\cdots+a_{k-d+4}q_{k-d+3}\\
&+q_{k-d+2}+|u_{r_{k-d+3}}|.
\end{aligned}
\]
All intermediate terms therefore cancel, leaving
\[
|R_k|-|u_{r_k}|=|s_{k-1}|-|u_{r_{k-d+3}}|.
\tag{3.5}
\]
The proper-prefix relation $u_{r_j}\prec s_j$ yields a positive right-hand side, hence
\(|u_{r_k}|<|Z_m|.\)
Together with $Z_m=R_k\prec s_k$ and $u_{r_k+1}=s_kt_k$, this gives
\(|u_{r_k}|<|Z_m|<|u_{r_k+1}|,\)
so $u_{r_k+1}=U_m$ is the required container.

Finally, if $d=3$ and $e_1=1$, the directive prefix ends in the next $0$-run. Writing $h=3m-2$, the same explicit image formula \cite[Lemma~5.9(v)]{Peltomaki2024} gives $Z_m=s_{h-1}$. By the generalized-standard-word relations \cite[Lemmas~3.15 and~3.19]{Peltomaki2024}, we have $u_{r_{h-1}}\prec Z_m$, while $Z_m$ is a prefix of $u_{r_{h-1}+1}$. Moreover, $u_{r_{h-1}+1}=\Pal(\delta^{m-1})$. Hence this is the shortest palindromic prefix containing $Z_m$.
\end{proof}

The following elementary lemma is the dictionary that turns a short derived word into information about the original finite word. It tells us both where the corresponding palindromic prefix occurs and which letters lie immediately to its left and right, so later context checks can be reduced to inspecting a few adjacent return types.

\begin{lemma}[Finite return-boundary and context lemma]
\label{lem:return}
Let $u$ be a palindromic prefix of an epistandard word, with return words $\rho(a)$. Let $y=y_1\cdots y_q$ be a finite factor of the derived word and $T=\rho(y)$. If $T$ starts at a return boundary of $u$, then the internal return boundaries are precisely the starts
\(1+|\rho(y_1\cdots y_i)|\qquad(0\le i<q)\)
of the corresponding occurrences of $u$ in $T$. At the internal boundary between return types $y_i$ and $y_{i+1}$, the occurrence of $u$ has context $(y_i,y_{i+1})$.
\end{lemma}

\begin{proof}
Each return word begins at one occurrence of $u$ and ends immediately before the next. For a palindromic prefix, the complete return coded by $a$ is a palindrome beginning with $ua$ and therefore ending with $au$; this is the specialization to a palindromic prefix of the complete-return description of Justin and Vuillon \cite[Theorem~4.4]{JustinVuillon2000}. Hence the occurrence at the boundary between return types $a$ and $b$ is preceded by $a$ and followed by $b$.
\end{proof}

\begin{remark}[Arbitrary factors]
\label{rem:factor-context}
If $v$ is not itself palindromic and $u=fvg$ is its shortest palindromic container, Justin--Vuillon's return-word correspondence \cite[Corollary~4.1]{JustinVuillon2000}, in the form $f^{-1}\rho(a)f$, shows that the left context varies with the return type exactly when $f=\eps$, and the right context varies with the return type exactly when $g=\eps$. Otherwise the corresponding context is fixed inside $u$.
\end{remark}

\section{The Case $e_{d-1}\ge2$}
\label{sec:rge2}

When the final directive letter is repeated at least twice, $S_m$ contains at least two consecutive copies of the same block $A_m$. This creates a strong periodic structure: two long boundary occurrences will overlap, and their unique expanded overlap will rule out every further net occurrence.

Assume $r=e_z\ge2$ and put $a_m=|A_m|$. Define
\(B_m=A_m^{r-1}C_m.\)

The first step is to rule out hidden periodicity inside the repeated block $A_m$. If $A_m$ were itself a proper power, smaller shifts could imitate the visible $A_m$-block structure and the later occurrence argument would no longer be rigid.

\begin{lemma}[Primitivity of the repeated block]
\label{lem:primitive}
The word $A_m$ is primitive for every $m\ge1$.
\end{lemma}

\begin{proof}
Let $M$ be the incidence matrix of the morphism $\mu^{m-1}\nu$. Since the incidence matrix of every $L_a$ is unimodular, $M$ is unimodular as well. The Parikh vector of
\(A_m=\mu^{m-1}\nu(z)\)
is therefore $M\mathbf e_z$, where $\mathbf e_z$ is the unit vector corresponding to $z$. Its coordinates have greatest common divisor one: indeed, if an integer $q\ge2$ divided every coordinate of $M\mathbf e_z$, then, since $M^{-1}$ has integer entries, $q$ would divide every coordinate of
\(M^{-1}M\mathbf e_z=\mathbf e_z,\)
a contradiction. Hence the coordinates of the Parikh vector of $A_m$ have greatest common divisor one. If $A_m=W^q$ with $q\ge2$, every coordinate of its Parikh vector would be divisible by $q$, which is impossible.
\end{proof}

Once $A_m$ is known to be primitive, the visible block length $a_m$ must be the true basic period of $S_m$. This is the key rigidity statement: it prevents an occurrence from sliding to a nontrivial position inside one $A_m$ block.

\begin{lemma}[Least period]
\label{lem:minperiod}
For $r\ge2$, the least period of $S_m$ is $a_m$.
\end{lemma}

\begin{proof}
The decomposition $S_m=A_m^rC_m$ and $C_m\prec A_m$ show that $a_m$ is a period. Suppose that $p<a_m$ were another period. The Fine--Wilf theorem states that if a word has periods $p$ and $q$ and length at least $p+q-\gcd(p,q)$, then $\gcd(p,q)$ is also a period \cite{FineWilf1965}. Here $|S_m|\ge2a_m$, so the theorem applies to the periods $p$ and $a_m$ and implies that $\gcd(p,a_m)$ is a period of the prefix $A_m$. Since this number properly divides $a_m$, the prefix $A_m$ would have a period properly dividing its length and hence would be a proper power, contradicting \cref{lem:primitive}.
\end{proof}

Before locating $B_m$, we make explicit the elementary periodicity argument that will be used twice. It formalizes the intuitive fact that, in a sufficiently long word of period $a$, a long prefix cannot reappear at a nontrivial shift smaller than $a$ unless that shift is itself a period.

\begin{lemma}[Periodic propagation]
\label{lem:period-propagation}
Let $T$ be a word of length $n$ having period $a$, let $q\ge1$, and assume $n\ge(q+1)a$. If the prefix $T[1..n-qa]$ occurs again starting at position $s+1$ for some $0<s<a$, then $s$ is a period of $T$.
\end{lemma}

\begin{proof}
The second occurrence gives $T[i]=T[i+s]$ for $1\le i\le n-qa$. We show that the same equality holds for every $i\le n-s$. Fix $i>n-qa$. Let $k\ge1$ be the least integer such that $j=i-ka\le n-qa$. Since $i<n$ and $i-(n-qa)<qa$, we have $k\le q$. By minimality, $i-(k-1)a>n-qa$, and hence
\(j=i-ka>n-(q+1)a\ge0,\)
so $j\ge1$. Repeatedly using the period $a$ gives $T[i]=T[j]$ and $T[i+s]=T[j+s]$. Since $j\le n-qa$, the equality supplied by the shifted prefix gives $T[j]=T[j+s]$. Thus $T[i]=T[i+s]$. Therefore $T[i]=T[i+s]$ for all $1\le i\le n-s$, which is exactly the statement that $s$ is a period of $T$.
\end{proof}

The least-period statement now fixes the positions of the long factor $B_m$. The next lemma says that only the two obvious copies survive---the prefix copy and the copy shifted by one whole $A_m$ block---and these will become the two net occurrences in this branch.

\begin{lemma}[The two occurrences of $B_m$]
\label{lem:Bocc}
For every $r\ge2$,
\(\Occ_{S_m}(B_m)=\{1,a_m+1\}.\)
\end{lemma}

\begin{proof}
The two displayed occurrences are immediate from $S_m=A_mB_m$. Since $|B_m|=|S_m|-a_m$, any occurrence of $B_m$ starts at a shift $s$ with $0\le s\le a_m$. The shifts $0$ and $a_m$ give the two displayed copies. If there were another one, then $0<s<a_m$ and the prefix of $S_m$ of length $|S_m|-a_m$ would reappear at shift $s$. Because $|S_m|\ge2a_m$, \cref{lem:period-propagation} with $q=1$ would make $s$ a period of $S_m$, contradicting the minimality of $a_m$ from \cref{lem:minperiod}.
\end{proof}

The two occurrences in \cref{lem:Bocc} are net: one is a prefix, the other a suffix, and $B_m$ occurs exactly twice. Since $|B_m|>a_m$, they overlap and form an ONOC. Their overlap is
\(O_m=A_m^{r-2}C_m.\)

To prove that these two occurrences are the complete list, the ONOC criterion asks for one final local fact: the one-letter expansion of their overlap must occur nowhere else. The next lemma supplies exactly this uniqueness, with a separate return-word argument only in the binary borderline case.

\begin{lemma}[Unique expanded bridge]
\label{lem:bridge-rge2}
For every $r\ge2$, the one-character expansion of the unique BNSO of the ONOC $\{B_m@1,B_m@(a_m+1)\}$ occurs exactly once in $S_m$.
\end{lemma}

\begin{proof}
If $r\ge3$, we have
\(\Occ_{S_m}(O_m)=\{1,a_m+1,2a_m+1\}.\)
Indeed, $|O_m|=|S_m|-2a_m$, so every possible start has shift between $0$ and $2a_m$. The shifts $0,a_m,2a_m$ give the three displayed copies. Suppose that an additional occurrence starts at shift $s$. If $a_m<s<2a_m$, then the period $a_m$ shifts this whole occurrence one block to the left, producing an occurrence at shift $s-a_m\in(0,a_m)$. We may therefore assume $0<s<a_m$. Since $r\ge3$ gives $|S_m|\ge3a_m$, \cref{lem:period-propagation} with $q=2$ would make $s$ a period of $S_m$, contradicting \cref{lem:minperiod}. Thus no additional start exists. Only the middle occurrence of $O_m$ has both a left and a right extension inside $S_m$, so its one-character expansion is unique.

If $r=2$ and $d\ge3$, then $O_m=C_m$ and \cref{lem:Cocc} gives
\(\Occ_{S_m}(C_m)=\{1,a_m+1,2a_m+1\}.\)
Again only the middle occurrence has both extension characters inside $S_m$, so its one-character expansion is unique.

It remains to consider the binary case $d=2$ with $r=2$. Write $a=e_0$ and
\(\delta=0^a1^2,\qquad \mu=L_0^aL_1^2.\)
At the directive cut after $\delta^{m-1}$, the shortest palindromic prefix containing $C_m=\mu^{m-1}(0)$ yields return words $C_m$ and $D_m=\mu^{m-1}(1)$, and the finite return-word factorization of $S_m$ is coded by
\((0^a1)^2 0.\)
Among the return boundaries carrying $C_m$, exactly two are preceded by return type $1$: the internal overlap and the final suffix occurrence. The internal one has right extension $0$, whereas the suffix occurrence has no right extension. Therefore the expanded BNSO is
\(1C_m0\)
and occurs exactly once. For $m=1$ this is checked directly in $S_1=(0^a1)^20$: the expanded bridge is $101$ when $a=1$ and $100$ when $a\ge2$, and is unique in either case.
\end{proof}

We can now close the whole branch $e_{d-1}\ge2$. The two long boundary occurrences cover $S_m$, and the unique expanded bridge prevents any additional net occurrence from crossing between them.

\begin{theorem}[Two-occurrence branch]
\label{thm:rge2}
If $e_{d-1}\ge2$, then for every $m\ge1$,
\[
\#\NO(S_m)=2.
\]
The only net occurrences are $B_m@1$ and $B_m@(|A_m|+1)$.
\end{theorem}

\begin{proof}
The two occurrences are net and form an ONOC. By \cref{lem:bridge-rge2}, its expanded BNSO is unique. The ONOC completeness lemma of Guo and Kishi \cite{GuoKishi2025} excludes every further net occurrence.
\end{proof}

\section{The Case $e_{d-1}=1$}
\label{sec:req1}

When the last directive letter occurs only once, the two boundary occurrences no longer cover the whole approximant by themselves. A third, central palindromic occurrence is therefore needed, and most of the work in this section is devoted to locating it and proving that the two resulting overlaps have unique one-letter expansions.

Assume $e_z=1$. We first treat $d\ge3$.

Let $w_m$ be $\delta^m$ with its final letters $z-1,z$ removed, let $\lambda'_m$ be its directive morphism, and put
\(X_m=\lambda'_m(z-1),\qquad Y_m=\lambda'_m(z),\qquad Z_m=\lambda'_m(0).\)
Since $\mu^m=\lambda'_mL_{z-1}L_z$ and
\(L_{z-1}L_z(0)=(z-1)z(z-1)0,\)
we have
\(S_m=X_mY_mX_mZ_m.\)
Moreover,
\(A_m=X_mY_m,\qquad C_m=X_mZ_m.\)

Define the palindromic prefix
\(P_m=\Pal(w_m),\)
which will provide the central net occurrence. Separately, let
\(H_m\)
denote the shortest palindromic prefix containing $Z_m$; see \cref{lem:localization}.

In the generic cases $H_m=P_m$. When $d=3$ and $e_1=1$, however, one has
\(H_m=\Pal(\delta^{m-1}) \qquad\text{and in general }H_m\ne P_m.\)
This distinction is essential: $P_m$ is the central net factor, whereas $H_m$ is used only to control occurrences of $Z_m$.

At the directive cut $w_m$, the finite derived prefix of $S_m$ is
\((z-1)z(z-1)0.\)

The central palindrome is useful only if its occurrences are completely controlled. The next lemma shows that exactly three complete copies of $P_m$ fit inside $S_m$ and, crucially, records their three different contexts; the middle copy will become the third net occurrence.

\begin{lemma}[Three occurrences of the central palindrome]
\label{lem:Pocc}
For $m\ge2$, $P_m$ has exactly three complete occurrences in $S_m$, starting at
\(1,\qquad |X_m|+1,\qquad |X_mY_m|+1,\)
with contexts
\((\bot,z-1),\qquad(z-1,z),\qquad(z,z-1).\)
\end{lemma}

\begin{proof}
Apply \cref{lem:return} at the directive cut $w_m$. The first three return boundaries of the finite derived word $(z-1)z(z-1)0$ give the three displayed occurrences and contexts. At the fourth boundary only $Z_m$ remains. In the generic case $H_m=P_m$ and \cref{lem:localization} gives $|Z_m|<|P_m|$. In the exceptional ternary case, $Z_m$ is a prefix of $H_m$ and $H_m\prec P_m$, so again $|Z_m|<|P_m|$. Hence no fourth complete occurrence of $P_m$ fits in $S_m$.
\end{proof}

The middle occurrence of $P_m$ is therefore net. Since $d\ge3$, \cref{lem:Cocc} with $r=1$ gives exactly two occurrences of $C_m$,
\(C_m@1,\qquad C_m@(|A_m|+1),\)
and the return-boundary contexts show that both are net. These three occurrences form an ONOC. Figure~\ref{fig:three-onoc} records the geometry used in the rest of the argument: the first and middle occurrences overlap in $Z_m$, while the middle and final occurrences overlap in $V_m$.

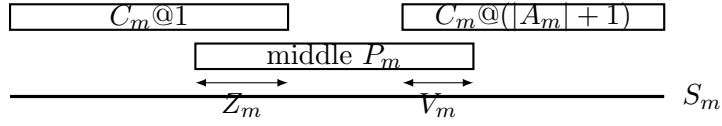
\begin{figure}[t]
\centering
\begin{tikzpicture}[x=0.72cm,y=0.52cm,>=latex]
  \draw[very thick] (0,0) -- (12,0);
  \node[anchor=west] at (12.15,0) {$S_m$};
  \draw[thick] (0,1.7) rectangle (5.1,2.35);
  \node at (2.55,2.03) {$C_m@1$};
  \draw[thick] (3.4,0.7) rectangle (8.5,1.35);
  \node at (5.95,1.03) {middle $P_m$};
  \draw[thick] (7.2,1.7) rectangle (12,2.35);
  \node at (9.6,2.03) {$C_m@(|A_m|+1)$};
  \draw[<->] (3.4,0.35) -- (5.1,0.35);
  \node[below] at (4.25,0.35) {$Z_m$};
  \draw[<->] (7.2,0.35) -- (8.5,0.35);
  \node[below] at (7.85,0.35) {$V_m$};
\end{tikzpicture}
\caption{Schematic geometry of the three-member ONOC in the branch $e_{d-1}=1$, $d\ge3$ (not to scale). The proof reduces completeness to uniqueness of the one-letter expansions of the two displayed overlaps.}
\label{fig:three-onoc}
\end{figure}

The first overlap is $Z_m$. Its expanded bridge is
\(E_m^{(1)}=(z-1)Z_m0.\)
We prove its uniqueness using the shortest palindromic container $H_m$ from \cref{lem:localization}. In the Justin--Vuillon return-word formula $H_m=fZ_mg$, we have $f=\eps$ because $Z_m$ is a prefix of $H_m$. Hence the left context of an occurrence of $Z_m$ at a return boundary is determined by the return type at that boundary. The occurrence of $Z_m$ forming the first overlap is the suffix of the prefix occurrence of $C_m=X_mZ_m$. Its left extension is the last letter of $X_m$, namely $z-1$. Moreover, $A_m=C_m\,\mu^{m-1}(0^{e_0-1}z)$ by \cref{lem:structural}; since $m\ge2$ and every letter image under $\mu$ begins with $0$, the suffix of $A_m$ following this prefix copy of $C_m$ begins with $0$. Hence this overlap occurrence is followed by $0$.

In the generic cases $H_m=P_m$. At the cut $w_m$ the finite derived prefix is
\((z-1)z(z-1)0.\)
Thus exactly two occurrences of $Z_m$ are preceded by return type $z-1$: the overlap occurrence just described and the terminal suffix occurrence. The latter has no right extension in $S_m$. Hence only the overlap occurrence realizes $(z-1)Z_m0$.

It remains to check the exceptional case $d=3$ and $e_1=1$. Write $a=e_0$. Here
\(H_m=\Pal(\delta^{m-1}),\)
and the return-word morphism at this cut is $\mu^{m-1}$. Consequently the finite derived word of $S_m$ is
\(\mu(0)=0^a1\,0^a2\,0^a1\,0.\)
Again there are exactly two boundaries preceded by return type $1=z-1$. The first gives the internal occurrence of $Z_m$ followed by $0$, while the second gives the terminal suffix occurrence and therefore has no right extension. Hence $(z-1)Z_m0$ occurs exactly once in $S_m$ also in the exceptional ternary case.

For the second overlap, let $q'_m$ be $\delta^{m-1}$ with its final $z$ deleted and put
\(V_m=\Pal(q'_m).\)
Then $w_m=q'_mzt_m$ for a $z$-free word $t_m$. We claim that
\[
P_m=Y_mV_m.
\tag{5.1}
\]
Indeed, the last-occurrence formula \cite{Justin2005} gives
\(\Pal(w_mz)=P_mV_m^{-1}P_m,\)
whereas Justin's formula \cite[Sec.~6.2.1]{GlenJustin2009} gives
\[
\Pal(w_mz)=\mu_{w_m}(z)P_m=Y_mP_m.
\]
Cancelling the common suffix $P_m$ yields $Y_m=P_mV_m^{-1}$, which is equivalent to (5.1).

It follows that the second and third members of the ONOC overlap exactly in $V_m$. To determine both extension letters and prove uniqueness, let $\kappa_m$ be the directive morphism associated with $q'_m$. Since $e_z=1$,
\[
\mu^m=\kappa_m L_z\mu,
\qquad\text{hence}\qquad
S_m=\kappa_m\bigl(L_z(\mu(0))\bigr).
\tag{5.2}
\]
Thus, at the directive cut $q'_m$, the finite derived word of $S_m$ with respect to the palindromic prefix $V_m$ is
\(\mathcal{D}_m=L_z(\mu(0)).\)
Now $\mu(0)=AC$ because $e_z=1$. The morphism $\nu$ uses only $L_0,\ldots,L_{z-1}$, so $A=\nu(z)$ contains exactly one occurrence of $z$, as its final letter, while $C=\nu(0)$ contains no $z$. Hence $\mu(0)$ contains exactly one $z$, and this $z$ is followed by a letter different from $z$. Under $L_z$, this unique transition produces one occurrence of $zz$; no other letter image, nor any boundary between two consecutive letter images, can produce $zz$. Therefore the two-letter factor $zz$ occurs exactly once in $\mathcal{D}_m$. By \cref{lem:return}, the unique $zz$-boundary corresponds to the unique occurrence of $V_m$ having context $(z,z)$. Consequently the expanded second BNSO is
\(E_m^{(2)}=zV_mz,\)
and it occurs exactly once in $S_m$.

At this point the geometry of the $e_{d-1}=1$ branch is complete: the two boundary copies of $C_m$ and the middle copy of $P_m$ form an ONOC, and the two expanded overlaps have just been shown to be unique. The following theorem packages these local facts into the global count.

\begin{theorem}[Three-occurrence branch]
\label{thm:r1-dge3}
If $d\ge3$, $e_{d-1}=1$, and $m\ge2$, then
\[
\#\NO(S_m)=3.
\]
\end{theorem}

\begin{proof}
The two boundary occurrences of $C_m$ and the middle occurrence of $P_m$ are net and form an ONOC. Both expanded BNSOs are unique, so the ONOC completeness lemma excludes all further net occurrences.
\end{proof}

\subsection{The binary Sturmian case}

For a binary alphabet, epistandard words coincide with the classical standard Sturmian words \cite{GlenJustin2009,Peltomaki2024}; we use this traditional name only to identify the class, and no additional Sturmian theory is needed below. The final count is the same as for $d\ge3$, but the localization mechanism used above is no longer valid. We therefore rebuild the three occurrences from a short return-word factorization and first make explicit the palindromic-containment facts that determine the two overlaps.

Let $d=2$, $a=e_0$, and $e_1=1$. Then $\mu=L_0^aL_1$, while $C_m=\mu^{m-1}(0)$ and $D_m=\mu^{m-1}(1)$. Moreover, $S_m=C_m^aD_mC_m$. For $a=1$, $\mu$ is the square of the usual Fibonacci morphism $\varphi$, defined by $\varphi(0)=01$ and $\varphi(1)=0$. Hence the result for $m\ge2$ follows from the Fibonacci theorem of Guo and Kishi \cite[Theorem~34]{GuoKishi2025}. Assume from now on that $a\ge2$.

For $j\ge1$, put
\(w_j=(0^a1)^{j-1}, \qquad \widehat H_j=\Pal(w_j).\)
The next lemma supplies the localization that replaces \cref{lem:localization} in the binary case. Besides identifying the shortest palindromic container of $C_m$, it shows directly that the two factors later called overlaps are exactly the geometric intersections of the three candidate net occurrences.

\begin{lemma}[Binary palindromic container and overlaps]
\label{lem:binary-container}
Let $m\ge2$ and $a\ge2$. Define
\[
\widehat B_m=\Pal(w_m0^{a-1}),
\qquad
O_m^{(1)}=\Pal(w_m0^{a-2}),
\qquad
O_m^{(2)}=\Pal(w_{m-1}0^a).
\]
Then:
\begin{enumerate}[label=(\roman*)]
\item $\widehat H_m$ is the shortest palindromic prefix containing $C_m$;
\item $\widehat B_m=C_m^{a-1}\widehat H_m$ and $O_m^{(1)}=C_m^{a-2}\widehat H_m$;
\item $\widehat H_m=D_mO_m^{(2)}$ and $O_m^{(2)}\prec C_m$.
\end{enumerate}
Consequently, the two occurrences of $\widehat B_m$ starting at $1$ and $|C_m|+1$ overlap exactly in $O_m^{(1)}$, while the second of these occurrences and the suffix occurrence of $C_m$ overlap exactly in $O_m^{(2)}$.
\end{lemma}

\begin{proof}
Justin's formula for iterated palindromic closure \cite[Sec.~6.2.1]{GlenJustin2009} gives, for every $j\ge1$ and $k\ge0$,
\(\Pal(w_j0^k)=C_j^k\widehat H_j,\)
because $\mu_{w_j}(0)=\mu^{j-1}(0)=C_j$. This immediately gives (ii).

For (iii), write $w_m=w_{m-1}0^a1$. Applying Justin's formula to the final $1$ yields $\widehat H_m=\mu_{w_{m-1}0^a}(1)\Pal(w_{m-1}0^a)=D_mO_m^{(2)}$, since $\mu_{w_{m-1}0^a}=\mu^{m-2}L_0^a$ and $\mu^{m-2}L_0^a(1)=\mu^{m-1}(1)=D_m$. Moreover, $C_m=\mu^{m-2}(0^a10)=D_mC_{m-1}$. We claim that $\widehat H_j\prec D_jC_j$ for every $j\ge1$. This is clear for $j=1$, since $\widehat H_1=\eps$. If it holds for $j$, then $O_{j+1}^{(2)}=C_j^a\widehat H_j\prec C_j^aD_jC_j=C_{j+1}$; consequently, $\widehat H_{j+1}=D_{j+1}O_{j+1}^{(2)}\prec D_{j+1}C_{j+1}$.
Thus the claim follows by induction, and in particular $O_m^{(2)}\prec C_m$, proving (iii).

The identities $C_m=D_mC_{m-1}$ and $\widehat H_m=D_mO_m^{(2)}$, together with the fact that $O_m^{(2)}$ begins with $C_{m-1}$, show that $C_m$ is a prefix of $\widehat H_m$. The palindromic prefix immediately preceding $\widehat H_m=\Pal(w_m)$ in the directive construction is $\Pal(w_{m-1}0^a)=O_m^{(2)}$, which is shorter than $C_m$ by (iii). Hence $\widehat H_m$ is the first, and therefore the shortest, palindromic prefix containing $C_m$. This proves (i).

Finally, by (ii), shifting the occurrence $\widehat B_m=C_m^{a-1}\widehat H_m$ by one $C_m$ block leaves the overlap $C_m^{a-2}\widehat H_m=O_m^{(1)}$. By (iii), the second occurrence of $\widehat B_m$ ends with $D_mO_m^{(2)}$, while the suffix copy of $C_m$ in $S_m=C_m^aD_mC_m$ starts immediately after $C_m^aD_m$ and has prefix $O_m^{(2)}$. Their intersection is therefore exactly $O_m^{(2)}$.
\end{proof}

At the directive cut $w_m0^{a-1}$ the return-word morphism is $\rho_m=\mu^{m-1}L_0^{a-1}$, and
\(S_m=\rho_m(010),\)
because $L_0^{a-1}(010)=0^a10=\mu(0)$. Hence the first two return boundaries give exactly the complete occurrences
\(\widehat B_m@1, \qquad \widehat B_m@(|C_m|+1),\)
with contexts $(\bot,0)$ and $(0,1)$, respectively; at the third boundary only $C_m$ remains, which is shorter than $\widehat B_m$. Thus both occurrences are net.

By \cref{lem:binary-container}(i), $\widehat H_m$ is the shortest palindromic container of $C_m$. Justin--Vuillon's occurrence correspondence therefore applies at the cut $w_m$, whose finite derived word in $S_m$ is $0^a10$. The only return boundary carrying $C_m$ that is preceded by return type $1$ is the boundary before the final suffix copy. Hence the suffix occurrence
\(C_m@(|S_m|-|C_m|+1)\)
has a unique left extension, while its right extension is unique by the boundary convention; it is therefore net. The three displayed net occurrences cover $S_m$ and form an ONOC.

By \cref{lem:binary-container}, its first overlap is $O_m^{(1)}$. At the corresponding directive cut,
\[
S_m=\mu^{m-1}L_0^{a-2}(0010),
\]
so the finite derived word is $0010$. The digram $00$ occurs there exactly once; by \cref{lem:return}, this gives the unique occurrence of $O_m^{(1)}$ with context $(0,0)$. Hence the expanded bridge $0O_m^{(1)}0$ is unique.

Similarly, the second overlap is $O_m^{(2)}$. At its directive cut,
\[
S_m=\mu^{m-2}L_0^a\bigl((10)^a110\bigr),
\]
and the digram $11$ occurs exactly once in the finite derived word $(10)^a110$. Thus $O_m^{(2)}$ has a unique occurrence with context $(1,1)$, and the expanded bridge $1O_m^{(2)}1$ is unique. The ONOC completeness lemma now gives
\(\#\NO(S_m)=3\qquad(m\ge2).\)

\subsection{The initial approximant and the main theorem}

The preceding arguments cover all noninitial approximants, but the first approximant can be too short for the same covering geometry. Fortunately its structure is simple enough to analyze directly: the binary case reduces to powers of $0$, while for $d\ge3$ the unique final letter $z$ sharply separates the possible repeated factors.

If $e_z\ge2$, \cref{thm:rge2} already gives $\#\NO(S_1)=2$. Suppose from now on that $e_z=1$.

We first dispose of the binary case $d=2$. Put $a=e_0$. Then
\(S_1=L_0^aL_1(0)=0^a10.\)
If $a=1$, the two occurrences of $0$ in $010$ are net, so $\#\NO(S_1)=2$. If $a=2$, all three occurrences of $0$ in $0010$ are net. If $a\ge3$, the two occurrences of $0^{a-1}$ starting at positions $1$ and $2$, together with the final occurrence of $0$, are net. No other power $0^k$ is net: for $k<a-1$ at least one one-letter extension is still repeated, while $0^a$ is unique. Since the letter $1$ occurs only once, no repeated factor can contain it. Therefore
\[
\#\NO(S_1)=
\begin{cases}
2,&a=1,\\
3,&a\ge2.
\end{cases}
\]

Assume henceforth that $d\ge3$. By \cref{lem:structural},
\(S_1=C0^{e_0-1}zC.\)
If $e_0=1$, then $S_1=CzC$. The letter $z$ occurs only once, because $C=\nu(0)$ and $\nu$ uses no $L_z$. Thus every factor crossing $z$ is unique. By \cref{lem:Cocc} with $r=1$, the two boundary copies of $C$ are the only occurrences of $C$ and both are net. Any further repeated occurrence is contained in one of these two copies and is therefore a proper sub-occurrence of a net occurrence, hence cannot itself be net. Thus $\#\NO(S_1)=2$.

It remains to consider $a=e_0\ge2$. Put $e=e_{z-1}$ and $\theta=L_0^aL_1^{e_1}\cdots L_{z-2}^{e_{z-2}}$; set $B=\theta(z-1)$ and $D=\theta(0)$. Since $\lambda'_1=\theta L_{z-1}^{e-1}$, we have $X_1=B$, $Z_1=B^{e-1}D$, and $C=B^eD=X_1Z_1$.
We next identify the central palindrome without using borderline generalized-standard-word indexing.
The last occurrence of $0$ in $w_1$ is the last letter of its initial run $0^a$.
The prefix preceding that occurrence is $0^{a-1}$.
The standard last-occurrence formula for iterated palindromic closure \cite{Justin2005} therefore gives
\(\Pal(w_10)=P_1\,(0^{a-1})^{-1}P_1.\)
On the other hand, Justin's formula \cite[Sec.~6.2.1]{GlenJustin2009}, together with $Z_1=\mu_{w_1}(0)$, gives
\(\Pal(w_10)=Z_1P_1.\)
Cancelling the common suffix $P_1$ yields
\[
P_1=Z_1 0^{a-1}.
\tag{5.6}
\]
Equivalently, using $A=X_1Y_1=C0^{a-1}z$ from \cref{lem:structural},
\(Y_1=P_1z.\)
Hence the finite derived prefix $(z-1)z(z-1)0$ at the cut defining $P_1$ yields three complete occurrences of $P_1$ with contexts
\((\bot,z-1),\qquad(z-1,z),\qquad(z,z-1),\)
because the fourth boundary leaves only $Z_1$, and $|Z_1|<|P_1|$ by (5.6). The middle occurrence, starting at $|X_1|+1$, is therefore net. Together with the two net boundary occurrences of $C$ supplied by \cref{lem:Cocc}, this gives three net occurrences in $S_1$.

We finally prove completeness. Let
\(U=C0^{a-1},\)
so that $S_1=UzC$. The prefix occurrence of $C$ and the middle occurrence of $P_1$ overlap exactly in $Z_1$ and together cover $U$. Their expanded overlap is
\[
E=(z-1)Z_1 0.
\tag{5.7}
\]
The word $B=\theta(z-1)$ contains exactly one $z-1$, as its final letter, whereas $D=\theta(0)$ contains no $z-1$. From (5.5), the factor $(z-1)Z_1$ can therefore start in the prefix copy of $C=B^eD$ only at the final letter of its first block $B$; in the suffix copy of $C$ the corresponding occurrence is terminal and has no following letter. Since $a\ge2$, the prefix occurrence is followed by $0$. Thus $E$ occurs exactly once in $S_1$.

Now let $N$ be any additional net occurrence. It cannot cross the unique letter $z$, since then its factor would be unique. If $N$ lies to the right of $z$, it is a proper sub-occurrence of the suffix net occurrence of $C$ and is impossible. Suppose that $N$ lies in $U$. If it is contained in either the prefix occurrence of $C$ or the middle occurrence of $P_1$, it is again a proper sub-occurrence of a net occurrence. Otherwise, because these two occurrences overlap and cover $U$, $N$ must start before the middle occurrence of $P_1$ and end after the prefix occurrence of $C$; hence it contains the expanded overlap $E$. Since $E$ is unique in $S_1$, every occurrence containing $E$ is unique, contradicting that $N$ is repeated. Therefore no additional net occurrence exists, and
\[
\#\NO(S_1)=3
\qquad(d\ge3,\ e_z=1,\ e_0\ge2).
\tag{5.8}
\]

All parameter regimes have now been exhausted. The following theorem simply assembles the two structural branches and the direct analysis of the initial approximant into a single exact classification.

\begin{theorem}[Main theorem]
\label{thm:main}
Let $d\ge2$ and $e_0,\ldots,e_{d-1}\ge1$. For
\(\mu=L_0^{e_0}\cdots L_{d-1}^{e_{d-1}},\qquad S_m=\mu^m(0),\)
one has
\[
\#\NO(S_m)=
\begin{cases}
3,&e_{d-1}=1\text{ and }(m\ge2\text{ or }e_0\ge2),\\
2,&\text{otherwise}.
\end{cases}
\]
In particular, for every $m\ge2$,
\[
\#\NO(S_m)=
\begin{cases}
3,&e_{d-1}=1,\\
2,&e_{d-1}\ge2.
\end{cases}
\]
\end{theorem}

\begin{proof}
If $e_{d-1}\ge2$, the statement follows from \cref{thm:rge2} for every $m\ge1$. Suppose $e_{d-1}=1$. For $m\ge2$ and $d\ge3$, \cref{thm:r1-dge3} gives exactly three net occurrences; for $d=2$, the binary Sturmian analysis of the preceding subsection gives the same conclusion. It remains only the initial approximant $m=1$: the preceding analysis of $S_1$ gives two net occurrences when $e_0=1$ and three when $e_0\ge2$. Combining these mutually exhaustive cases proves both displayed formulas.
\end{proof}

The $d$-bonacci family is obtained by setting every partial quotient equal to one. In that case one full directive period is exactly $d$ substitution steps, so the main theorem immediately transfers to the period-aligned $d$-bonacci approximants.

\begin{corollary}[Period-aligned $d$-bonacci approximants]
Let $\sigma$ be the standard $d$-bonacci substitution, defined by $\sigma(i)=0(i+1)$ for $0\le i<d-1$ and $\sigma(d-1)=0$. If $e_0=\cdots=e_{d-1}=1$, then $\mu=\sigma^d$ and hence
\[
\#\NO\bigl(\sigma^{dm}(0)\bigr)=3
\qquad(m\ge2).
\]
\end{corollary}

\section{Consequences and Outlook}

We have completely determined the net occurrences in the finite morphic approximants generated by periodic regular epistandard directive words of the form
\(\Delta=\left(0^{e_0}1^{e_1}\cdots(d-1)^{e_{d-1}}\right)^\omega.\)
For every noninitial approximant, the number of net occurrences exhibits the sharp dichotomy
\(\#\NO(S_m)= \begin{cases} 3,&e_{d-1}=1,\\ 2,&e_{d-1}\ge2, \end{cases} \qquad m\ge2.\)
Thus, despite the arbitrarily large alphabet and arbitrary positive exponents in one directive period, the number of net occurrences is controlled by a single parameter: the exponent of the last directive letter.

The two cases have distinct combinatorial geometries. When $e_{d-1}\ge2$, two canonical boundary occurrences overlap and form an ONOC. When $e_{d-1}=1$, a palindromic factor determined by the directive structure contributes a third, central net occurrence. In both cases, return words identify the relevant occurrences and their contexts, while uniqueness of the expanded BNSOs excludes every additional net occurrence.

The period-aligned $d$-bonacci consequence shows that the same three-occurrence geometry persists along the subsequence $\sigma^{dm}(0)$. The remaining residue classes $\sigma^{dm+r}(0)$, $1\le r<d$, are not consequences of the present theorem and form a concrete phase-dependent extension of the problem. More generally, the proof illustrates how return-word techniques and ONOCs complement each other: derived words turn context questions into short finite patterns, whereas ONOCs turn local uniqueness into global completeness.

A natural next question is whether the periodicity assumption on the partial quotients can be removed. General regular epistandard words retain cyclic directive letters but may have arbitrary positive partial quotients \cite{Peltomaki2024}. It would be interesting to determine whether an analogue of the two-or-three dichotomy survives in that wider setting, and more broadly which morphic or more general S-adic (iterated-morphism) systems admit a uniformly bounded number of net occurrences in their canonical finite approximants.

\section*{Declaration on generative AI assistance}
During the preparation of this work, the authors used ChatGPT (OpenAI) to support manuscript organization, language editing, and the checking and presentation of mathematical arguments. After using this tool, the authors reviewed and edited the content as needed and take full responsibility for the content of the article.

\end{document}